\documentclass[11pt,english]{article}
\usepackage[T1]{fontenc}
\usepackage[latin9]{inputenc}
\usepackage{float}
\usepackage{units}
\usepackage{textcomp}
\usepackage{amsmath}
\usepackage{amssymb}
\usepackage{graphicx}
\usepackage{algorithm}
\usepackage[noend]{algorithmic}
\usepackage{caption}
\usepackage{subcaption}

\makeatletter

\floatstyle{ruled}
\newfloat{algorithm}{tbp}{loa}
\providecommand{\algorithmname}{Algorithm}
\floatname{algorithm}{\protect\algorithmname}

\usepackage{latexsym}\@ifundefined{definecolor}
 {\usepackage{color}}{}
\usepackage{mathtools}\usepackage{amsthm}\usepackage{epsfig}\usepackage{algorithmic}\usepackage{wasysym}\usepackage{wrapfig}\usepackage{appendix}\usepackage{fancyhdr}\usepackage{enumerate}
\usepackage{times,euscript}

\usepackage{fullpage}

\def\reals{{\mathbb R}}
\def\eps{{\varepsilon}}
\def\A{{\cal A}}

\def\D{{\cal D}}

\def\Q{{\cal Q}}

\def\bd{{\partial}}

\def\graph{\EuScript{G}}
\def\dfrechet{\delta^*}

\def\@{\,}

\newtheorem{theorem}{Theorem}[section]
\newtheorem{lemma}[theorem]{Lemma}
\newtheorem{corollary}[theorem]{Corollary}
\newtheorem{observation}[theorem]{Observation}

\def\paragraph#1{\medskip\noindent\textbf{#1}}

\makeatother
\date{}
\usepackage{babel}
\begin{document}
\begin{titlepage}
\title{A faster algorithm for the discrete Fr\'echet distance under translation\thanks{
Work by Haim Kaplan has been supported by Israel Science Foundation grant no. 822/10, and
German-Israeli Foundation for Scientific Research and Development (GIF) grant no. 1161/2011, and Israeli Centers of Research Excellence (I-CORE) program (Center  No. 4/11).
Work by Micha Sharir has been supported 
by Grant 892/13 from the Israel Science Foundation, 
by the Israeli Centers of Research Excellence (I-CORE)
program (Center No.~4/11),
and by the Hermann Minkowski--MINERVA Center for Geometry at Tel Aviv
University.
Work by Micha Sharir and Rinat Ben Avraham has been supported 
by Grant 2012/229 from the U.S.-Israeli Binational Science Foundation.
}}

\author{Rinat Ben Avraham%
\thanks{School of Computer Science, Tel Aviv University, Tel Aviv 69978, Israel;
\texttt{rinatba@gmail.com}%
} \and Haim Kaplan%
\thanks{School of Computer Science, Tel Aviv University, Tel Aviv 69978, Israel;
\texttt{haimk@post.tau.ac.il}%
} \and Micha Sharir%
\thanks{School of Computer Science, Tel Aviv University, Tel~Aviv 69978,
Israel; \texttt{michas@post.tau.ac.il }%
} }
\maketitle

\begin{abstract}
The discrete Fr\'echet distance is a useful similarity measure for comparing two sequences of points $P=(p_1,\ldots, p_m)$ and $Q=(q_1,\ldots,q_n)$. In many applications, the quality of the matching can be improved if we let $Q$ undergo some transformation relative to $P$. In this paper we consider the problem of finding a translation of $Q$ that brings the discrete Fr\'echet distance between $P$ and $Q$ to a minimum. We devise an algorithm that computes the minimum discrete Fr\'echet distance under translation in $\mathbb{R}^2$, and runs in $O(m^3n^2(1+\log(n/m))\log(m+n))$ time, assuming $m\leq n$. This improves a previous algorithm of Jiang et al.~\cite{JXZ08}, which runs in $O(m^3n^3 \log(m + n))$ time.

\end{abstract}
\end{titlepage}

\section{Introduction}
\label{sec:introduction}
Shape matching is an important area of computational geometry, that has
applications in computer vision, pattern recognition, and other
fields that are concerned with matching objects by shape similarity.
Generally, in shape matching we are given two geometric objects $A$ and $B$ and we want to measure to what extent they are similar.  Usually we may allow certain transformations,
like translations, rotations and/or scalings, of one object relative to the other, in order to improve the quality of the match.

In many applications, the input data consists of finite sets of points sampled from the actual objects.
To measure similarity between the sampled point sets,
 various distance functions have been used.  One popular function  is the Hausdorff distance
that equals to
 the maximum distance from a point in one set to
 its nearest
point in the other set.
However, when the objects which we compare are curves, sequences,
 or contours of larger objects, and the sampled points are ordered along the
compared contours,
 the discrete Fr\'echet
distance may be a more appropriate similarity measure. This is because the discrete Fr\'echet
distance takes into account the ordering of the
points along the contours which the Hausdorff distance ignores.
 Comparing curves and sequences is a major task that arises in computer vision, image processing and bioinformatics (e.g., in matching backbone sequences of proteins).

The \emph{discrete Fr\'echet distance} between a sequence of points $P$ and another sequence of points $Q$ is defined as the minimum, over all possible independent (forward) traversals of the sequences, of the maximum distance between the current point of $P$ and the current point of $Q$ during the traversals. See below and in Section~\ref{sec:preliminaries} for a more formal definition.

In this work, we focus on the problem of computing the minimum discrete Fr\'echet distance \emph{under translation}. That is, given two sequences $P$ and $Q$ of $m$ and $n$ points, respectively, in the plane, we wish to translate $Q$ by a vector $t\in \mathbb{R}^2$ such that the discrete Fr\'echet distance between $P$ and $Q+t$ is minimized.

\medskip
\noindent\textbf{Background.}
The Fr\'echet distance has been extensively studied during the past 20 years. The main variant,
the continuous Fr\'echet distance, where no transformation is allowed, measures similarity between (polygonal) curves. It is the smallest $\delta$ for which there exist forward simultaneous traversals of the two curves, from start to end, so that at all times the distance between the corresponding points on the curves is at most $\delta$. The discrete Fr\'echet distance considers sequences $P$ and $Q$ of points instead of curves.
It is defined analogously, where (a) the simultaneous traversals of the sequences are represented as a sequence of pairs $(p^{(1)}, q^{(1)}),\ldots,(p^{(t)},q^{(t)})$, where $p^{(i)}\in P$, $q^{(i)}\in Q$, for $i=1,\ldots,t$, (b) the first (resp., last) pair consists of the starting (resp., terminal) points of the two sequences, and (c) each $(p^{(i)}, q^{(i)})$ is obtained from $(p^{(i-1)}, q^{(i-1)})$ by moving one (or both) point(s) to the next position in the corresponding sequence.
Most studies of the problem consider the situation where no translation (or any other transformation) is allowed. In this ``stationary'' case, the discrete Fr\'echet distance in
the plane can be computed, using dynamic programming, in
$O(mn)$ time (Eiter and
Mannila~\cite{EM94}). Agarwal et al.~\cite{ABKS12} slightly improve this bound, and show
that the (stationary) discrete Fr\'echet distance can be computed in
$O\left(\dfrac{mn\log\log n}{\log n}\right)$ time on a word RAM, and a very recent result of Bringmann~\cite{Bri14} indicates that a substantially subquadratic solution (one that runs in time $O((mn)^{1-\delta})$, for some $\delta>0$) is unlikely to exist.
Alt and Godau~\cite{AG95} showed that the (stationary) continuous Fr\'echet distance of two
planar polygonal curves with $m$ and $n$ edges, respectively, can be computed,
using dynamic programming, in $O(mn\log mn)$ time.  This has been slightly improved recently by Buchin et al.~\cite{BBMM12}, who
showed that the continuous
Fr\'echet distance can be computed in $O(N^2 (\log N)^{1/2} (\log\log N)^{3/2})$
time on a pointer machine, and in $O(N^2 (\log\log N)^2)$ time
on a word RAM (here $N=m=n$ denotes the number of edges in each curve).
In short, the best known algorithms for the stationary case, for both discrete and continuous variants, hover around the quadratic time bound.

Not surprisingly, the problems become much harder, and their solutions much less efficient, when translations (or other transformations) are allowed.
For the problem of computing the minimum continuous Fr\'echet distance under translation, Alt et al.~\cite{AKW01} give an algorithm with
$O(m^3n^3(m+n)^2\log(m+n))$ running time, where $m$ and $n$ are the number of edges in the curves. They also give a $(1+\eps)$-approximation algorithm for the problem, that runs in $O(\eps^{-2}mn)$ time. That is, they compute a translation of one of the curves relative to the other, such that the Fr\'echet distance between the resulting curves is at most $(1+\eps)$ times the minimum Fr\'echet distance under any translation.
In three dimensions, Wenk~\cite{Wenk02} showed that, given two polygonal chains with $m$ and $n$ edges respectively, the
minimum continuous Fr\'echet distance between them, under any reasonable family of transformations, can be computed in $O((m+n)^{3f+2} \log (m+n))$ time,
where $f$ is the number of degrees of freedom for moving one chain relative to the other. So with translations alone $(f=3)$,
the minimum continuous Fr\'echet distance in $\mathbb{R}^3$ can be computed in $O((m+n)^{11} \log (m+n))$ time, and when both
translations and rotations are allowed $(f=6)$, the corresponding minimum continuous Fr\'echet distance can
be computed in $O((m+n)^{20} \log (m+n))$ time.

The situation with the discrete Fr\'echet distance under translation is somewhat better, albeit still inefficient.
Jiang et al.~\cite{JXZ08} show
that, given two sequences of points in the plane, the minimum discrete Fr\'echet distance between
them under translation can be computed
in $O(m^3n^3 \log(m + n))$
time. For the case where both rotations and translations are allowed, they give an algorithm that runs in $O(m^4n^4 \log(m + n))$ time.
They also design a heuristic method for aligning two sequences of points under translation and rotation in three dimensions.
Mosig et al.~\cite{MC05} present an approximation algorithm
that computes the
 discrete Fr\'echet distance under translation, rotation
and scaling in the plane, up to a factor close to $2$, and runs in $O(m^2n^2)$ time.

\medskip
\noindent\textbf{Our results.} Our algorithm improves the bound of Jiang et al.~\cite{JXZ08} by a nearly linear factor, with running time $O(m^3n^2(1+\log(n/m))\log(m+n))$, assuming $m\leq n$. It uses a $0/1$-matrix $M(P,Q)$ of size $m\times n$, whose rows (resp., columns) correspond to the points of $P$  (resp., of $Q$). Assuming a stationary situation, or, rather, a fixed translation of $Q$, an entry in the matrix is equal to 1 if and only if the distance between the two corresponding points is at most $\delta$, where $\delta$ is some fixed distance threshold. We use $(i,j)$ to denote an entry in the matrix that corresponds to the points $p_i$ and $q_j$, and we use $M_{i,j}$ to denote its value. The discrete Fr\'echet distance is at most $\delta$ if and only if there is a row- and column-monotone path of ones in $M$ that starts at $(1,1)$ and ends at $(m,n)$ (see Section~\ref{sec:preliminaries} for a more precise definition).

We can partition the plane of translations into a subdivision $\A_\delta$ with $O(m^2n^2)$ regions, so that, for all translations in the same region, the matrix $M$ is fixed (for the fixed $\delta$). We then traverse the regions of $\A_\delta$, moving at each step from one region to a neighboring one. Assuming general position, in each step of our traversal exactly one entry of $M$ changes from $1$ to $0$ or vice versa.
We present a dynamic data structure $\Gamma(M)$ that supports an update of an entry of $M$, in $O(m(1+\log(n/m)))$ time, assuming $m\leq n$,\footnote{This is without loss of generality as we can change the roles of $m$ and $n$ by flipping $M$.} and then re-determines whether there is a monotone path of ones from $(1,1)$ to $(m,n)$, in $O(1)$ additional time.
If we find such a monotone path in $M$, we have found a translation $t$ (actually a whole region of translations\footnote{For a critical value of $\delta$, the region can degenerate to a single vertex of $\A_\delta$; see Sections~\ref{sec:arrangement} and~\ref{sec:optimization} for details.}) such that the discrete Fr\'echet distance between $P$ and $Q+t$ is at most $\delta$. Otherwise, when we traverse the entire $\A_\delta$ and fail after each update, we conclude that no such translation exists. Using this procedure, combined with the parametric searching technique~\cite{NM83}, we obtain an algorithm for computing the minimum discrete Fr\'echet distance under translation.

We reduce the dynamic maintenance of $M$  to dynamic maintenance of reachability in a planar graph, as edges are inserted and deleted to/from the graph.
Specifically, we can think of (the 1-entries of) $M$ as a representation of a planar directed graph with $N\leq mn$ nodes. Each 1-entry of $M$ corresponds to a node in the graph, and each possible forward move in a joint traversal is represented by an edge (see Section~\ref{sec:preliminaries} for details). Then, determining whether there is a row- and column-monotone path of ones from $(1,1)$ to $(m,n)$ corresponds to a reachability query in the graph (from $(1,1)$ to $(m,n)$).

A data structure for dynamic maintenance of reachability in directed planar graphs was given by Subramanian~\cite{Sub93}. This data structure supports updates and reachability queries in $O(N^{2/3}\log N)$ time, where $N$ is the number of nodes in the graph.
Diks and Sankowski~\cite{DS07} improved this data structure, and gave a structure that supports updates and reachability queries in $O(N^{1/2}\log^3 N)$ time.

We give a simpler and more efficient structure for maintaining reachability in $M$ that exploits its special structure.
 Our structure can update reachability information in $M$ in $O(m(1+\log(n/m)))$ time, assuming $m\leq n$,
and answers reachability query (from $(1,1)$ to $(m,n)$) in $O(1)$ time. In contrast, the data structure of
  \cite{DS07} applied in our context performs an update and a query in $O((mn)^{1/2}\log^3 (mn))$ time.
Using our structure, we obtain an algorithm for computing the minimum discrete Fr\'echet distance under translation that runs in $O(m^3n^2(1 +\log(n/m))\log(m+n))$ time (again, assuming $m\leq n$).

To summarize the contributions of this paper
are  twofold: (a) The reduction of the problem of computing the minimum discrete Fr\'echet distance to a dynamic planar directed graph reachability problem. (b) An efficient data structure for this reachability problem. For $m\approx n$ our structure is faster than the general reachability structure of \cite{DS07} by a polylogarithmic factor, and when $m\ll n$ the improvement is considerably more significant (roughly by a factor $\sqrt{n/m}$). Moreover, our data structure is simpler than that of Diks and Sankowski.


\section{Preliminaries}
\label{sec:preliminaries}
We now define the (stationary) discrete Fr\'echet distance  formally. Let $P=(p_1,\ldots,p_{m})$ and
$Q=(q_1,\ldots,q_{n})$ be the two planar sequences defined in the introduction.

For some fixed distance $\delta>0$ we
 define a $0/1$-matrix $M_{\delta}(P,Q)$ formally as follows.  The rows (resp., columns) of $M_\delta(P,Q)$ correspond to the points
of $P$ (resp., of $Q$) in their given order. An entry $(i,j)$
of $M_\delta(P,Q)$ is  1 if the distance between $p_i$ and $q_j$ is at most $\delta$, and is 0 otherwise.
we denote $M_\delta(P,Q)$ by $M$ when $P$ and $Q$ and $\delta$ are clear from the context.

The directed graph $\graph_\delta(P,Q)$ associated with $P$, $Q$ and $\delta$ has a vertex for each
pair $(p_i,q_j)\in P\times Q$ and an edge for each pair of adjacent ones in $M_\delta(P,Q)$.
Specifically, we have an edge from $(p_i,q_j)$ to $(p_{i+1},q_j)$ if and only if both $(i,j)$ and $(i+1,j)$
are 1 in $M$, an edge from $(p_i,q_j)$ to $(p_{i},q_{j+1})$ if and only if both $(i,j)$ and $(i,j+1)$
are 1 in $M$, and an edge from $(p_i,q_j)$ to $(p_{i+1},q_{j+1})$ if and only if both $(i,j)$ and $(i+1,j+1)$
are 1 in $M$. we denote $G_\delta(P,Q)$ by $G$ when $P$ and $Q$ and $\delta$ are clear from the context.

The \emph{(stationary) discrete Fr\'echet distance}
between $P$ and $Q$, denoted by $\dfrechet(P,Q)$, is
the smallest $\delta>0$ for which
$(p_{m},q_{n})$ is reachable from $(p_1,q_1)$ in $\graph_\delta$. Informally,
think of $P$ and $Q$ as two sequences of stepping stones and of two
frogs, the $P$-frog and the $Q$-frog, where the $P$-frog has to
visit all the $P$-stones in order and the $Q$-frog has to visit all
the $Q$-stones in order. The frogs are connected by a
rope of length $\delta$, and are initially placed at $p_1$ and
$q_1$, respectively. At each move, either one of the frogs jumps
from its current stone to the next one and the other stays at its current stone,
or both of them jump  simultaneously
from their current stones to the next ones. Furthermore, such a jump is allowed only if
the distances between the two frogs before and after the
jump are both at most $\delta$.
Then $\dfrechet(P,Q)$ is the smallest $\delta>0$ for which there exists
a sequence of jumps that gets the frogs to $p_{m}$ and $q_{n}$,
respectively.

The problem of computing the minimum discrete Fr\'echet distance under translation, as reviewed in the introduction, is to find a translation $t$ such that $\delta^*(P,Q+t)$ is minimized.

We say that an entry $(i,j)$ of $M$ is \emph{reachable} from an entry $(k,l)$, with $k\leq i, l\leq j$, if $(p_i,q_j)$ is reachable from $(p_k,q_l)$ in $\graph$.
A path from $(p_k,q_l)$ to $(p_i,q_j)$ in $\graph$ corresponds to a (weakly) row-monotone and column-monotone sequence of
ones in  $M$ connecting  the one in entry $(k,l)$ to the one in entry $(i,j)$.
This is sequence consists of three kinds of moves: 1) {\em upward moves} between entries of the form $(r,s)$ to $(r+1,s)$
in which the $P$-frog moves
from $p_{r}$ to $p_{r+1}$, 2) {\em right moves} between entries of the form $(r,s)$ to $(r,s+1)$
in which the $P$-frog moves
from $q_{s}$ to $q_{s+1}$, and 3) {\em diagonal moves} between entries of the form $(r,s)$ to $(r+1,s+1)$
in which the $P$-frog moves
from $q_{s}$ to $q_{s+1}$ both frogs move simultaneously --- the $P$-frog from $p_{r}$ to $p_{r+1}$, and the $Q$-frog from
$q_{s}$ to $q_{s+1}$.
 See
Figure~\ref{fig:staircase}.
We call such a monotone sequence of ones in $M$ a {\em path in M} from $(k,l)$ to $(i,j)$.
To determine whether
$\dfrechet(P,Q)\le\delta$, we need to determine whether there is
such a  path in $M$ that starts at
$(1,1)$ and ends at $(m,n)$.
We say that an entry $(i,j)$ of $M$ is \emph{reachable} if there is a path from $(1,1)$ to $(i,j)$.

We denote the concatenation of two  paths $\pi_1,\pi_2$ by $\pi_1\cdot\pi_2$, assuming that the last entry of $\pi_1$ is the first entry of $\pi_2$; this entry appears only once in the concatenation.

\begin{figure}[htb]

\centering\begin{tabular}{cc}
\includegraphics[scale=0.8]{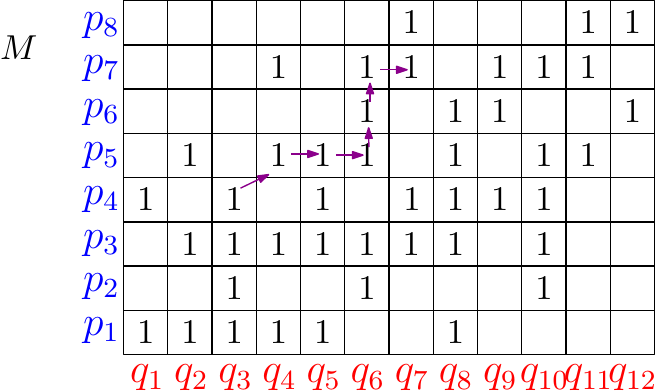} & \hspace{0.5cm}
\includegraphics[scale=0.8]{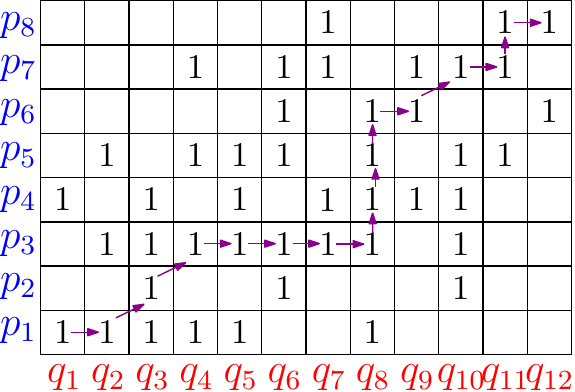} \\
(a) & (b)
\end{tabular}

\centering \caption{\small (a) A reachability path from $(4,3)$ to $(7,7)$. (b) A reachability path from $(1,1)$ to $(8,12)$.}
\label{fig:staircase}
\end{figure}


We use  a decomposition of $M$ into \emph{blocks}. A block is a submatrix of $M$ that corresponds to contiguous subsequences $P'$ and $Q'$ of $P$ and $Q$, respectively.
We denote by $M(P', Q')$ the block of $M(P,Q)$ formed by $P'$ and $Q'$. Consider a block $M(P', Q')$ of $M(P,Q)$. Let $p^-$ (resp., $p^+$) denote the first (resp., last) point of $P'$ and let $q^-$ (resp., $q^+$) denote the first (resp., last) point of $Q'$. We call the entries of $M(P', Q')$ corresponding to $\{p^-\} \times Q'\cup P'  \times \{q^-\}$ the \emph{input boundary} of $M(P', Q')$, and denote it by $M(P', Q')^-$ (the common entry corresponding to $\{p^-\}\times\{q^-\}$ appears only once in the boundary). Similarly, we call the entries of $M(P', Q')$ corresponding to $\{p^+\} \times Q' \cup P' \times \{q^+\}$ the \emph{output boundary} of $M(P', Q')$, and denote it by $M(P', Q')^+$ (with a similar suppression of the duplication of the common element $\{p^+\}\times\{q^+\}$). Note that there is a two-entry overlap between the input and output boundaries. We enumerate the entries of $M(P', Q')^-$ by first enumerating the entries of $\{p^-\}\times Q'$ from right to left (i.e., backwards) and then the remaining entries of $P' \times \{q^-\}$ from bottom to top (forward). We enumerate the entries of $M(P', Q')^+$ by first enumerating the entries of $P' \times \{q^+\}$ from bottom to top (forward) and then the remaining entries of $\{p\}^+ \times Q'$ from right to left (backwards). Informally, $M(P',Q')^-$ is enumerated in ``clockwise'' order, while $M(P',Q')^+$ is enumerated in ``counterclockwise'' order; see Figure~\ref{fig:block}. For two entries $i,j$ of this enumeration of an input or output boundary $B$ of $M(P', Q')$, we use $[i,j]$ to denote the sequence of entries $(i,i+1,\ldots,j)$ of $B$.

\begin{figure}[htb]


%


\centering\includegraphics[scale=0.8]{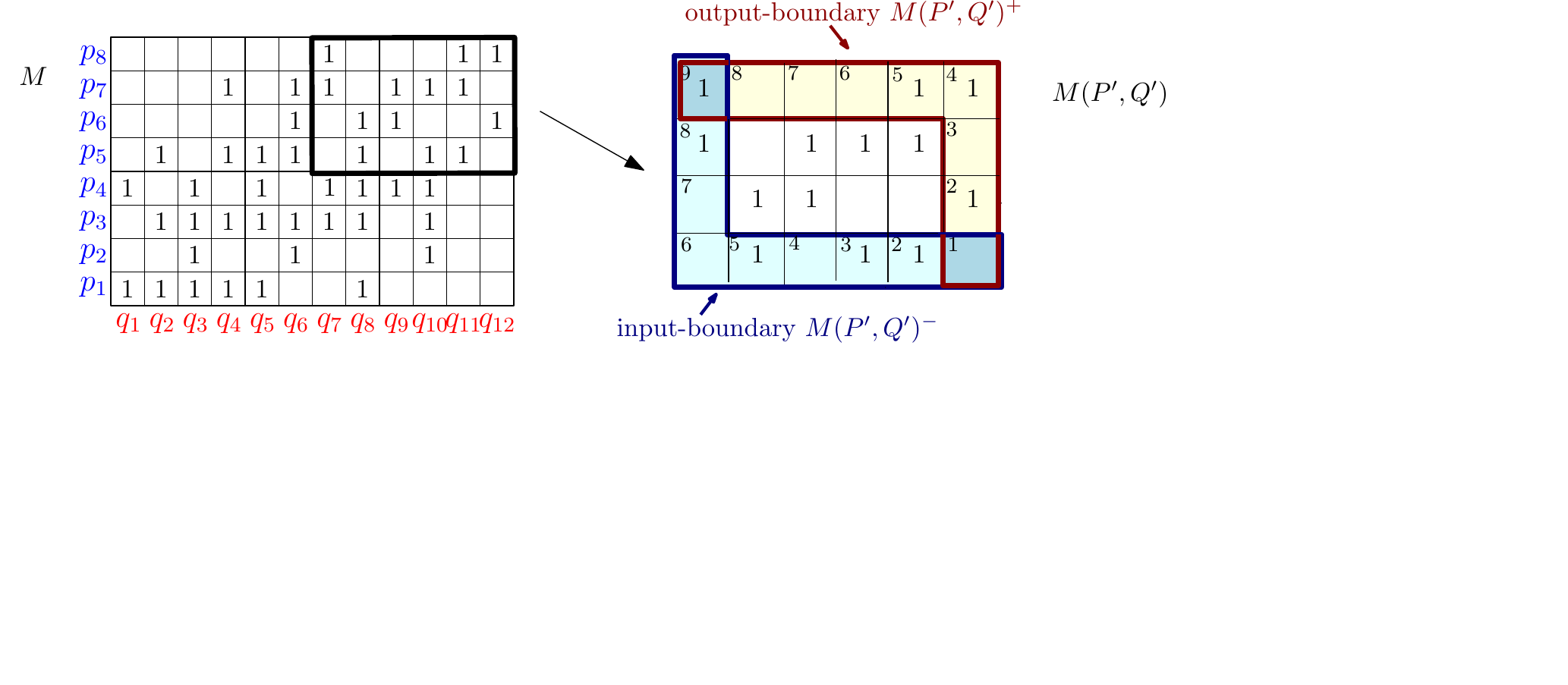}
\vspace{-4.1cm}

\hspace{-1.5cm} (a) \hspace{5.2cm} (b)

\centering\caption{\small (a) A (highlighted) block $M(P',Q')$ of $M(P,Q)$ , where $P'=(p_5,p_6, p_7, p_8)$ and $Q' = (q_7,q_8,\ldots,q_{12})$. (b) The input boundary $M(P',Q')^-$ and the output boundary $M(P',Q')^+$ of $M(P',Q')$ are marked, with the orderings of their elements.}
\label{fig:block}
\end{figure}

We also use the following definitions. We call the entries corresponding to $P'  \times \{q^-\}$ the \emph{vertical input boundary} of $M(P',Q')$, and denote it by $\overline{M(P',Q')}\@^-$. We call the entries corresponding to $P'  \times \{q^+\}$ the \emph{vertical output boundary} of $M(P',Q')$, and denote it by $\overline{M(P',Q')}\@^+$. That is, $\overline{M(P',Q')}\@^-$ and $\overline{M(P',Q')}\@^+$ are the vertical parts of $M(P',Q')^-$ and $M(P',Q')^+$, respectively. We enumerate the entries of each vertical boundary from bottom to top.

\section{The subdivision of the plane of translations}
\label{sec:arrangement}
We first consider the corresponding decision problem.
That is, given a value $\delta>0$, we wish to decide whether there exists a translation $t\in \mathbb{R}^2$ such that $\dfrechet(P,Q+t)\le\delta$.

For a point $x\in \mathbb{R}^2$, let $D_\delta(x)$ be the disk of radius $\delta$ centered at $x$.
Given two points $p_i\in P$ and $q_j\in Q$, consider the disk $D_\delta(p_i-q_j)$, and notice that $t\in D_\delta(p_i-q_j)$ if and only if $\|(p_i-q_j)-t\| \leq \delta$ (or $\|p_i-(q_j+t)\| \leq \delta$). That is, $D_\delta(p_i-q_j)$ is precisely the set of translations $t$ for which $q_j+t$ is at distance at most $\delta$ from $p_i$.

We construct the arrangement $\A_\delta=\A_\delta(P,Q)$ of the disks in $\D=\{D_\delta(p_i-q_j)\mid (p_i,q_j)\in P\times Q\}$. We assume general position of the points. That is, we assume that (a) no more than two boundaries of these disks intersect in a common vertex of $\A_\delta$, and (b) no pair of the disks are tangent to each other. Nevertheless, such a degeneracy can arise when $\delta$ is a \emph{critical value} (see Section~\ref{sec:optimization} for details about critical values of $\delta$ that arise during the optimization procedure), but we assume that at most one such degeneracy can happen for a given $\delta$. Since the number of disks is $mn$, the combinatorial complexity of $\A_\delta$ is $O(m^2n^2)$. Let $f$ be a face of $\A_\delta$ of any dimension $0,1$ or $2$ (by convention, $f$ is assumed to be relatively open), and let $t\in f$ be a translation. Then, for points $p_i\in P, q_j \in Q$, $q_j+t$ is at distance at most $\delta$ from $p_i$ if and only if the disk $D_\delta(p_i-q_j)$ contains $f$ (otherwise, the disk is disjoint from $f$). Since this holds for every $t\in f$, it follows that
$f$ corresponds to a unique pairwise-distances matrix $M(P,Q+t)$, for any $t \in f$. We denote this matrix by $M(P, Q+f)$, for short.

The setup just described leads to the following naive solution for the decision problem. Construct the arrangement $\A_\delta$ for the given distance $\delta$, and traverse its faces. For each face $f\in \A_\delta$, form the corresponding pairwise-distances matrix $M(P,Q+f)$, and solve the (stationary) discrete Fr\'echet distance decision problem for $P$ and $Q+f$ using a straightforward dynamic programming on $M(P,Q+f)$ (or the more sophisticated slightly subquadratic algorithm of Agarwal et al.~\cite{ABKS12}). If $\delta^*(P, Q+f) \leq \delta$ for some face $f$, we conclude that there exists a translation $t$ such that $\delta^*(P, Q+t) \leq \delta$ (any translation $t\in f$ would do). If the entire arrangement $\A_\delta$ is traversed and no face $f$ of $\A_\delta$ satisfies $\delta^*(P, Q+f) \leq \delta$, we determine that $\delta^*(P, Q+t) > \delta$ for all translations $t\in \mathbb{R}^2$. The complexity of $\A_\delta$ is $O(m^2n^2)$, and solving the discrete Fr\'echet distance decision problem for each face of $\A_\delta$ takes $O(mn)$ time (or slightly less, as in~\cite{ABKS12}). Hence, the solution just described for the decision problem takes (slightly less than) $O(m^3n^3)$ time.

Jiang et al.~\cite{JXZ08} used an equivalent solution for the decision problem, that takes the same asymptotic running time. Rephrasing their procedure in terms of $\A_\delta$, they test whether $\delta^*(P,Q+t)\leq \delta$
for translations $t$  corresponding to
 the vertices of $\A_\delta$, and over an additional set of $mn$ translations, one chosen from the boundary of each disk. The correctness of this approach follows by observing that if $f$ is a face of $\A_\delta$ and $t$ is any point on $\bd{f}$, then all the $1$-entries of $M(P,Q+f)$ are also $1$-entries of $M(P,Q+t)$, so it suffices to test the vertices of $f$, or, if $f$ has no vertices, to test an arbitrary point $t\in \bd{f}$. We will use this observation in our implementation of the optimization procedure.

Our naive solution is similar to the algorithm of Jiang et al.~\cite{JXZ08}, in the sense that they both discretize the set of possible translations. However, our solution is more suitable for the improvement of this naive bound, that we present in Section~\ref{sec:dynamic}, since it allows us to traverse the set of possible translations in a manner that introduces only a single change in $M(P, Q+f)$, when we move from one face $f$ of translations to a neighboring one.

To exploit this property we need a data structure that maintains reachability data for $M$, and updates it efficiently after each change. We present this structure in two stages. First, in Section~\ref{sec:linear_reachability}, we present a compact reachability structure for blocks of $M(P,Q+f)$, which is the main building block of the overall structure. Then, in Section~\ref{sec:dynamic}, we present the overall data structure, and show how to use it to improve the naive solution sketched above by a nearly linear factor.

\section{Compact representation of reachability in a block}
\label{sec:linear_reachability}

Let $B$ be a block of $M=M(P,Q+f)$ of size $r\times c$, and suppose that we have already computed the reachable entries of $B^-$ and we then wish to compute the reachable entries of $B^+$. If the entries of the block are given explicitly, this can be done in $O(r c)$ time using dynamic programming (or slightly faster using the algorithm of~\cite{ABKS12}). Our goal in this section is to design a data structure, that we denote as $\Phi(B)$, that allows us to compute the reachable entries of $B^+$ from the reachable entries of $B^-$, in $O(r+c)$ time. The overall data structure itself is constructed recursively from these block structures (see Section~\ref{sec:dynamic} for details), and implicitly accesses all the entries of $B$.
The advantage of using this block decomposition is that updating the structure can be done more efficiently.

\vspace{-1cm}
\begin{figure}[htb]

\centering\begin{tabular}{cc}
\includegraphics[scale=0.6]{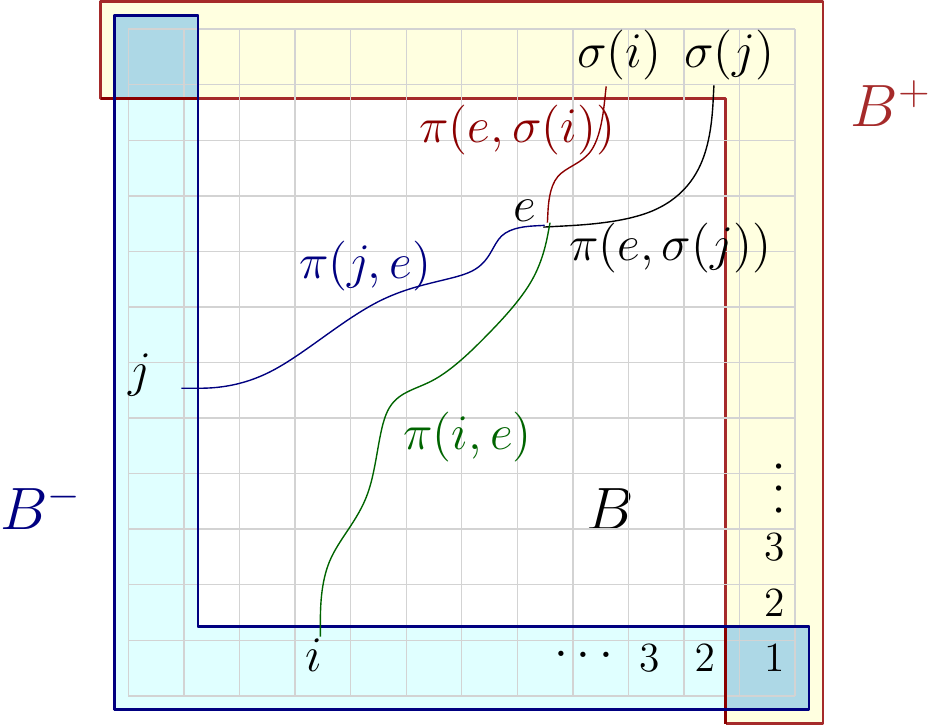} & 
\includegraphics[scale=0.6]{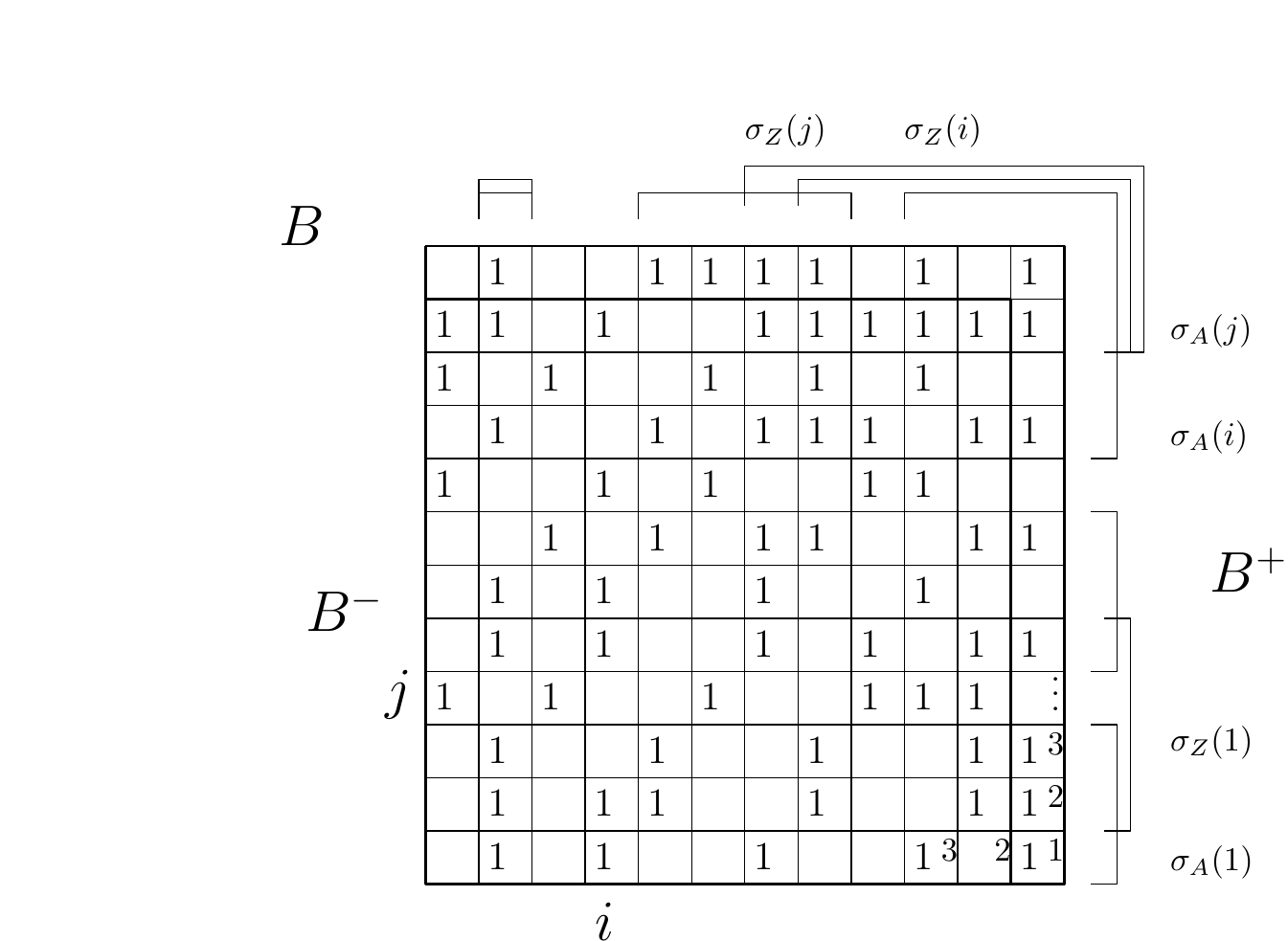} \\
(a) & (b)
\end{tabular}

\centering \caption{\small (a) Two entries $i,j$ of $B^-$ and two entries $\sigma(i)$ and $\sigma(j)$ of $B^+$ that are reachable from $i$ and $j$, respectively. Since $i<j$ and $\sigma(i)>\sigma(j)$, $\sigma(j)$ is reachable also from $i$, and $\sigma(i)$ is reachable also from $j$. (b) The intervals $[\sigma_A(k),\sigma_Z(k)]$, for any $1$-entry $k$ of $B^-$, are either disjoint or overlap in a common subinterval. Neither of these intervals can strictly contain both endpoints of the other. These intervals are defined (and shown in the figure) only for $1$-entries of $B^-$.}
\label{fig:united_fig}
\end{figure}


\begin{observation}\label{obs:monge}
Let $B$ be a block of $M$ and let $i,j$ be two entries of $B^-$ such that $j>i$ (in the ``clockwise'' order defined in Section~\ref{sec:preliminaries}).
Let $\sigma(i)$ be an entry of $B^+$ that is reachable from $i$, and let $\sigma(j)$ be an entry of $B^+$ that is reachable from $j$. If $\sigma(j) <\sigma(i)$ (in the corresponding ``counterclockwise'' order) then $\sigma(j)$ is also reachable from $i$, and $\sigma(i)$ is also reachable from $j$.
\end{observation}
\begin{proof}
See Figure~\ref{fig:united_fig}(a). Since $\sigma(i)$ is reachable from $i$, there is a (monotone)  path $\pi(i, \sigma(i))$ from $i$ to $\sigma(i)$ in $M$. Similarly, since $\sigma(j)$ is reachable from $j$, there is a (monotone) path $\pi(j, \sigma(j))$ from $j$ to $\sigma(j)$. Since $i<j$ and $\sigma(j) <\sigma(i)$, $\pi(i, \sigma(i))$ must cross $\pi(j, \sigma(j))$ (i.e., there exists a 1-entry $e \in \pi(i, \sigma(i)) \cap \pi(j, \sigma(j))$. Hence, $\pi(i, \sigma(i))$ can be decomposed into two subpaths $\pi(i,e), \pi(e, \sigma(i))$ such that $\pi(i, \sigma(i))= \pi(i,e)\cdot \pi(e, \sigma(i))$, and $\pi(j,\sigma(j))$ can be similarly decomposed as $\pi(j, \sigma(j))= \pi(j,e)\cdot \pi(e, \sigma(j))$.
As a result, the paths $\pi(i, \sigma(j)) = \pi(i,e) \cdot \pi(e, \sigma(j))$ and $\pi(j, \sigma(i)) = \pi(j,e) \cdot \pi(e, \sigma(i))$ are also (monotone)  paths, and the claim follows.
\end{proof}

\begin{corollary}\label{cor:reachable_interval}
Let $B$ be a block of $M$, let $i$ be an entry of $B^-$ and let $\sigma_1(i),\sigma_2(i)$ be two entries in $B^+$ that are both reachable from $i$, with $\sigma_1(i)<\sigma_2(i)$. If there exists an entry $\sigma(j)$ that is reachable from some $j\in B^-$, such that $\sigma_1(i)<\sigma(j)<\sigma_2(i)$, then $\sigma(j)$ is also reachable from $i$.
\end{corollary}
\begin{proof}
By Observation~\ref{obs:monge}, if $i<j$, then $\sigma(j)$ is reachable from $i$ since $\sigma(j)<\sigma_2(i)$. If $i>j$, then $\sigma(j)$ is reachable from $i$ since $\sigma_1(i)<\sigma(j)$.
\end{proof}

The corollary is applied as follows. Let $i$ be an entry of $B^-$, let $\sigma_A(i)$ and $\sigma_Z(i)$ denote the first and last entries in $B^+$ that are reachable from $i$. (Note that
for these entries to be defined, the value of the entry $i$ must be $1$. Symmetrically, the values of both $\sigma_A(i)$ and $\sigma_Z(i)$, if defined, must be equal to $1$.) Then the interval $[\sigma_A(i), \sigma_Z(i)]$ can only contain entries of the following three types.
\begin{enumerate}
\item
1-entries that are reachable from $i$.
\item
0-entries.
\item
1-entries that are not reachable from $i$, nor from any other entry of $B^-$.
\end{enumerate}

In other words, $[\sigma_A(i), \sigma_Z(i)]$ cannot contain $1$-entries that are reachable from some $j$ in $B^-$ and not from $i$.


\begin{corollary}\label{cor:overlapping_intervals}
Let $B$ be a block of $M$ and let $i$ and $j$ be two entries of $B^-$ such that $j>i$. Then $\sigma_A(j)\geq \sigma_A(i)$ and $\sigma_Z(j)\geq \sigma_Z(i)$.
\end{corollary}
\begin{proof}
Assume to the contrary that $\sigma_A(i)>\sigma_A(j)$. Then, according to Observation~\ref{obs:monge}, $\sigma_A(j)$ is reachable from $i$. Hence, $\sigma_A(i) \leq \sigma_A(j)$, a contradiction. Similarly, if $\sigma_Z(j)<\sigma_Z(i)$, then  Observation~\ref{obs:monge} implies that $\sigma_Z(i)$ is reachable from $j$. Hence, $\sigma_Z(j) \geq \sigma_Z(i)$, contradiction.
\end{proof}
In other words, $[\sigma_A(i),\sigma_Z(i)]$ and $[\sigma_A(j),\sigma_Z(j)]$ can be either disjoint or overlap in a common subinterval, but they cannot be properly nested inside one another (that is, neither of these intervals can contain both endpoints of the other in its ``interior''). Note, however, that one interval can weakly contain the other. That is, if one interval contains the other, then either $\sigma_A(i)=\sigma_A(j)$ or $\sigma_Z(i)=\sigma_Z(j)$, or both. See Figure~\ref{fig:united_fig}(b).

Let $B$ be a block of $M$ of size $r\times c$.
We construct a data structure $\Phi(B)$ for $B$, which stores the following information. (Here we only specify the structure; its construction is detailed in Section~\ref{sec:dynamic}.)
\begin{enumerate}
\item\label{enum:B^-}
For each $1$-entry $i$ of $B^-$ we store
\begin{enumerate}
\item
the first entry $\sigma_A(i)$ of $B^+$ that is reachable from $i$, and
\item
the last entry $\sigma_Z(i)$ of $B^+$ that is reachable from $i$.
\end{enumerate}
\item \label{enum:B^+}
For each $1$-entry $j$ of $B^+$ we store
\begin{enumerate}
\item
a flag $f(j)$ indicating whether $j$ is reachable from some entry of $B^-$.
\item
a list $L_A(j)$ of the $1$-entries $i\in B^-$ such that $\sigma_A(i)= j$, and
\item
a list $L_Z(j)$ of the $1$-entries $i\in B^-$ such that $\sigma_Z(i)= j$.
\end{enumerate}
\end{enumerate}
\begin{lemma}\label{lem:linear}
Given the data structure $\Phi(B)$ for a block $B$, and given the entries of $B^-$ that are reachable from $(1,1)$, we can determine, in $O(r+c)$ time, the entries of $B^+$ that are reachable from $(1,1)$.
\end{lemma}
\begin{proof}
We go over the reachable $1$-entries of $B^-$ in order. For each such entry $i$, we go over the entries in the interval $I(i)=[\max\{\sigma_Z(i^-),\sigma_A(i)\}, \sigma_Z(i)]$ of $B^+$, where $i^-$ is the previous reachable $1$-entry of $B^-$ (for the first reachable entry $i$ of $B^-$, $I(i)=[\sigma_A(i),\sigma_Z(i)]$ and $i^-$ is undefined).
Note that, by Corollary~\ref{cor:overlapping_intervals}, $\max\{\sigma_Z(i^-),\sigma_A(i)\}\in [\sigma_A(i), \sigma_Z(i)]$ so
$I(i)\subseteq[\sigma_A(i), \sigma_Z(i)]$.  (The entries of $[\sigma_A(i), \sigma_Z(i)]$ that precede $\max\{\sigma_Z(i^-),\sigma_A(i)\}$ were already processed when we went over $I(i^-)$ or over intervals associated with earlier indices.)

For each $1$-entry $j$ of $I(i)$ that is reachable from some entry of $B^-$ (according to the flag $f(j)$), we determine that $j$ is reachable also from $(1,1)$. Since we traverse each interval $[\sigma_A(i), \sigma_Z(i)]$ starting from $\max\{\sigma_Z(i^-),\sigma_A(i)\}$, the internal portions of the subintervals that we inspect are pairwise disjoint, implying that the running time is linear in $r+c$. We omit the straightforward proof of correctness of this procedure.
\end{proof}

\section{Dynamic maintenance of reachability in $M(P,Q+f)$}
\label{sec:dynamic}
We present a data structure, that uses the compact representation of reachability in a block of the previous section, to support an update of a single entry of $M$ in $O(m(1+\log(n/m)))$ time, assuming $m\leq n$.
We present this data structure in two stages. First, in Section~\ref{sec:data}, we show how to support an update of a single entry in $O(m)$ time, in the case where $M$ is a square matrix of size $m\times m$. Then, in Section~\ref{sec:improved}, we generalize this data structure to support an update of a single entry in $O(m(1+\log(n/m)))$ time, in the general case where $M$ is an $m\times n$ matrix with $m\leq n$ (the case $m\geq n$ is treated in a fully symmetric manner).

In Section~\ref{sec:overall}, we  describe the overall decision procedure that improves the naive solution sketched in Section~\ref{sec:arrangement}, using this dynamic data structure.

\subsection{A dynamic data structure for reachability maintenance in a square matrix}
\label{sec:data}

\begin{figure}[htb]
\vspace{-4cm}
\centering\begin{tabular}{cc}

\hspace{-1.2cm}\includegraphics[scale=0.6]{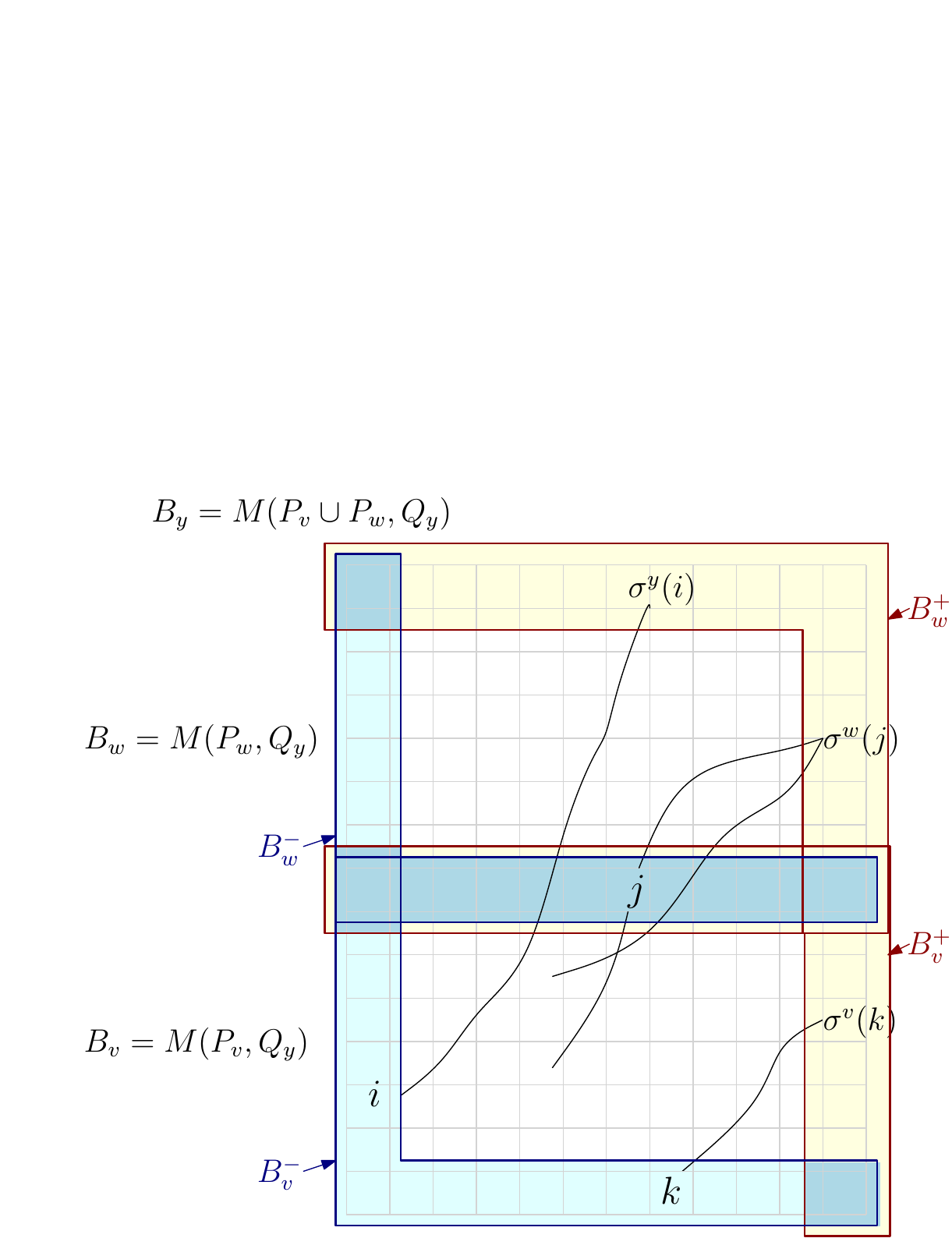} & 
\includegraphics[scale=0.6]{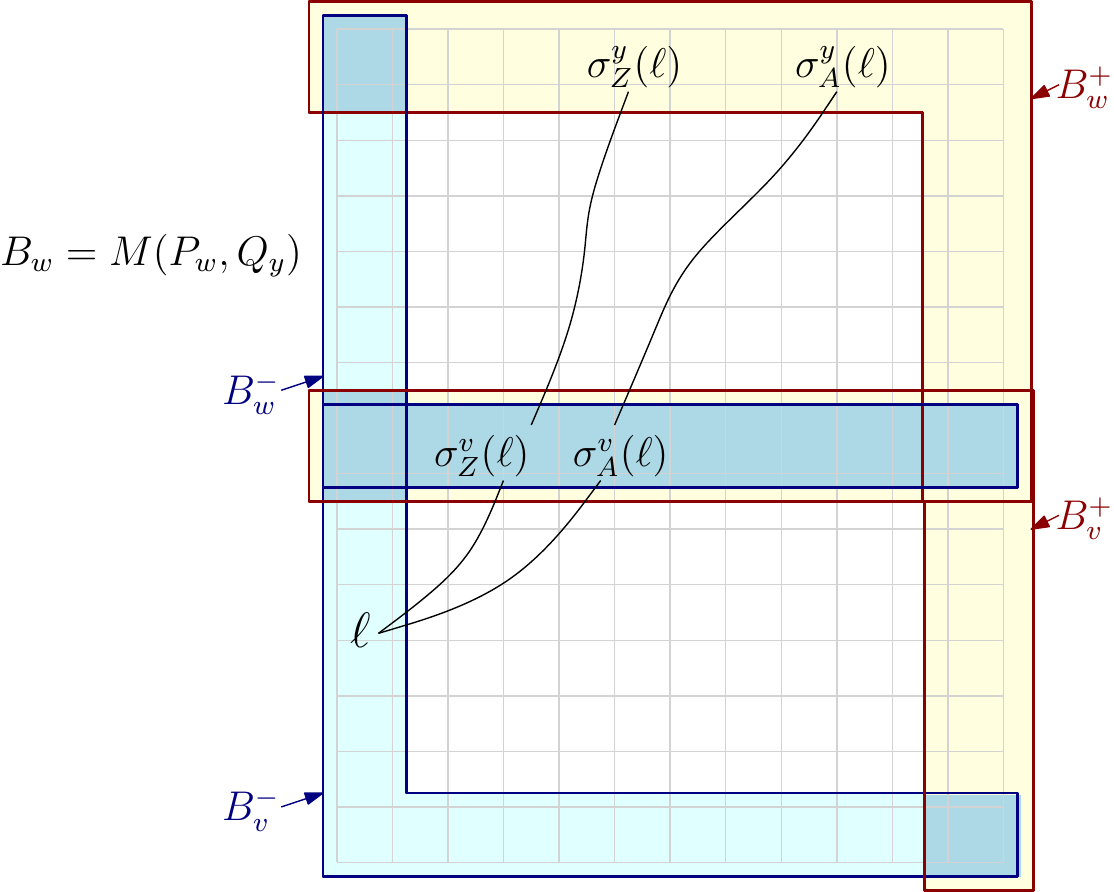} \\
(a)  & \hspace{2cm}(b)
\end{tabular}
\centering \caption{\small A block $B_y=M(P_y,Q_y)$, corresponding to a node $y$ of $\Gamma$, is composed of the blocks of the children $v,w$ of $y$. The block $B_v=M(P_v,Q_v)$ corresponding to the left child $v$ lies below the block $B_w=M(P_w,Q_w)$ that corresponds to the right child $w$, and we have $P_y=P_v\cup P_w$ and $Q_y=Q_v=Q_w$.
\\(a)
$\sigma^{y}(i)$ and $\sigma^{y}(k)$ are examples of reachable entries of $B_y^+$ and we have $\sigma^{y}(i)=\sigma^{w}(\sigma^{v}(i))$ and $\sigma^{y}(k)=\sigma^{v}(k)$. $\sigma^{w}(j)$ is an entry of $B_w^+$ that is reachable from $B_w^-$, but it is not a reachable entry of $B_y^+$ (from $B_y^-$) since all the paths in $B_y$ that lead to $\sigma^{w}(j)$ go through entries of $B_v^+$ that are not reachable from $B_v^-$.
\\(b)
$\sigma_A^{y}(\ell)=\sigma_A^{w}(\sigma_A^{v}(\ell))$ and $\sigma_Z^{y}(\ell)=\sigma_Z^{w}(\sigma_Z^{v}(\ell))$.
}
\label{fig:united_blocks}
\end{figure}

We store the reachability data of $M(P,Q+f)$ (of some arbitrary face $f$ from which we start the traversal of the arrangement $\A_\delta$) in a so-called \emph{decomposition tree} $\Gamma$, by halving $P$ and $Q$ alternately.
That is, the root $v$ of $\Gamma$ corresponds to the entire matrix $M(P,Q+f)$ and we store at $v$ the reachability information $\Phi(M(P,Q+f))$, as described in the previous section. (The actual construction of the reachability data, at all nodes of $\Gamma$, is done bottom-up, as described below.) In the next level of $\Gamma$ we partition $P$ into two subsequences $P_1,P_2$, of at most $\lfloor m/2\rfloor+1$ points each, such that the last point of $P_1$ is the first point of $P_2$, and obtain a corresponding ``horizontal'' partition of $M(P,Q+f)$ into two blocks $M(P_1,Q+f)$, $M(P_2,Q+f)$, each of size at most $(\lfloor m/2\rfloor+1) \times m$, with a common ``horizontal'' boundary. We create two children $v_1, v_2$ of $v$ and store at each $v_i$ the reachability information $\Phi(M(P_i,Q+f))$, for $i=1,2$. In the next level of $\Gamma$, we partition $Q$ into two subsequences $Q_1,Q_2$, of at most $\lfloor m/2\rfloor+1$ points each, such that the last point of $Q_1$ is the first point of $Q_2$, and obtain a corresponding ``vertical'' partition of each block $M(P_i,Q+f), i\in\{1,2\}$, into two blocks $M(P_i,Q_j+f), j\in\{1,2\}$, each of size at most $(\lfloor m/2\rfloor+1) \times (\lfloor m/2\rfloor+1)$, with a common vertical boundary. We construct four respective grandchildren, and store the corresponding reachability structures $\Phi(M(P_i,Q_j+f))$ at these nodes.
We continue recursively to partition each block by halving it horizontally or vertically, alternately, in the same manner, until we reach blocks of size $2\times 2$.
For each node $v$ of $\Gamma$, let $P_v$ and $Q_v$ denote the subsequences of $P$ and $Q$ that form the block $M(P_v, Q_v+f)$ that is associated with $v$. To simplify the notation, we denote $\Phi(M(P_v,Q_v+f))$ as $\Phi_v$, for each node $v$.

The reachability data $\Phi_v$ at the nodes $v$ of $\Gamma$ is computed by a bottom-up traversal of $\Gamma$, starting from the leaves. The construction of $\Phi(M(P_v,Q_v+f))$ at a leaf $v$ is trivial, and takes constant time. The following lemma provides an efficient procedure for constructing the reachability data at inner nodes of $\Gamma$.

\begin{lemma}\label{lem:union}
Let $y$ be an inner node of $\Gamma$ with left and right children $v$ and $w$, where the blocks stored at $v,w$ have a common horizontal boundary.
Given the reachability data $\Phi_v, \Phi_w$, the data $\Phi_y$ can be computed in $O(|P_y|+|Q_y|)$ time. An analogous statement holds when the common boundary of the children blocks is vertical.
\end{lemma}
\begin{proof}
Note that in the setup of the lemma, we have $Q_y=Q_v=Q_w$ and $P_y=P_v\cup P_w$. By construction, $M(P_v,Q_y)$ lies below $M(P_w,Q_y)$. Denote $M(P_v,Q_y)$ by $B_v$, $M(P_w,Q_y)$ by $B_w$, and $M(P_y,Q_y)$ by $B_y$. For each entry $i$ of $B_y^-$, denote by $\sigma_A^y(i)$ (resp., $\sigma_Z^y(i)$) the first (resp., last) entry of $B_y^+$ that is reachable from $i$. We also use
$\sigma^y(i)$ to denote an entry of $B_y^+$ that is reachable from $i$ in $B_y$. Analogous notations are used for the children blocks $B_v, B_w$.
See Figure~\ref{fig:united_blocks}.

We first copy the reachability information from the boundaries of $B_v$ and $B_w$ to the boundary of $B_y$ (except for the ``interior'' portion $B^*_{vw}$ of the common boundary $B_v^+ \cap B_w^-$ of $B_v$ and $B_w$, which is not a boundary of $B_y$). The data for the $1$-entries on the left boundary of $B_w$ (which are of type \ref{enum:B^-} in the definition of $\Phi$) is still valid, since the reachability paths of $B_y$ that start at these entries are fully contained in $B_w$.
Similarly, the data for the $1$-entries on the right boundary of $B_v$ (which are of type \ref{enum:B^+}) is still valid, since the reachability paths of $B_y$ that end at these entries are fully contained in $B_v$.
We thus need to determine the reachability information from the $1$-entries of the input boundary $B_v^-$ of $B_v$ to the entries of the output boundary $B_w^+$ of $B_w$, and merge it with the already available data, to get the complete structure $\Phi$ at $y$.

First note that an entry $j$ of $B_w^+$ that is reachable from $B_w^-$ may now become unreachable from $B_y^-$. This happens if all the reachability paths in $B_y$ to $j$ go through entries on $B^*_{vw}$ that are not reachable from $B_v^-$. See Figure~\ref{fig:united_blocks}(a).
We thus need to turn the flag $f(j)$ of such entries to false. To do this, we go over the entries of $B_w^+$ in order, and maintain a queue $\Q$ that satisfies the invariant that, when we are at an entry $j$ of $B_w^+$, $\Q$ contains all the entries $i$ of $B_w^-$ that are reachable from $B_y^-$, such that $j$ is reachable from $i$. That is, $\Q$ contains all the entries $i \in B^*_{vw}$ that are reachable from $B_v^-$ such that $j$ is reachable from $i$, and all the entries $i\in B_w^-\setminus B_{vw}^*$ (that is, the left side of $B_w^-$) such that $j$ is reachable from $i$. We start with an empty queue. For each $1$-entry $j$ of $B_w^+$ we first go over the list $L_A(j)$ (of $\Phi_w$), and for each element $i$ in $L_A(j)$ that is in $B^*_{vw}$, we check if it is reachable from $B_v^-$ (using the flag $f(i)$ from $\Phi_v$). If it is, we put it in $\Q$. We also add to $\Q$ each element in $L_A(j)$ that is in $B_w^-\setminus B_{vw}^*$. If $\Q$ is empty, there is no reachability path from $B_y^-$ to $j$ and we set $f(j)$ to be false. We then go over the list $L_Z(j)$ (of $\Phi_w$) and remove from $\Q$ each element in $L_Z(j)$ that is in $\Q$. This traversal takes $O(|P_y|+|Q_y|)$ time, since each element of $B_w^-$ appears at most once in the lists $L_A$ and at most once in the lists $L_Z$.
The correctness follows from the invariant that when we go over an entry $j\in B_w^+$, all the entries of $B_w^-$ that $j$ is reachable from, and that are reachable from $B_y^-$, are in $\Q$. The invariant is maintained correctly because each time that an interval $[\sigma_A(i),\sigma_Z(i)]$ of an entry $i\in B_w^-$ begins (and $i$ is reachable from $B_y^-$), $i$ is inserted into $\Q$, and when the interval ends, $i$ is removed from $\Q$, so $i$ is in $\Q$ for all entries $j$ that are reachable from $i$. In conclusion, if $\Q$ is empty, $j$ is not reachable from $B_y^-$ and the flag $f(j)$ can be turned false. Otherwise, $i$ is reachable from $B_y^-$.

We now update the intervals $[\sigma_A(i),\sigma_Z(i)]$ of the entries $i\in B_v^-$ and, in correspondence, the lists $L_A(\sigma(i)), L_Z(\sigma(i))$ of $B_w^+$ (where $\sigma(i)$ is any entry of $B_w^+$ that is reachable from $i$).
Consider a $1$-entry $i$ of $B_v^-$ and consider an entry $\sigma^v(i)$ in $B^*_{vw}$; that is, $\sigma^v(i)$ is a 1-entry in $[\sigma_A^v(i),\sigma_Z^v(i)]$ that is reachable from $i$. By transitivity, the entries $\sigma^w(\sigma^v(i))$ of $B_w^+$ that are reachable from $\sigma^v(i)$ are also reachable from $i$. We update $[\sigma_A(i),\sigma_Z(i)]$ according to this rule, as follows (see Figure~\ref{fig:united_blocks}(b)). We set $\sigma_A^{y}(i)=\sigma_A^{w}(\sigma_A^{v}(i))$, for each entry $i\in B_v^-$ such that $\sigma_A^v(i)\in B^*_{vw}$; correspondingly, we also add $i$ to $L_A(\sigma_A^{y}(i))$.
Similarly, for each entry $i\in B_v^-$ such that $\sigma^v_Z(i)\in B^*_{vw}$, we set $\sigma_Z^{y}(i)=\sigma_Z^{w}(\sigma_Z^{v}(i))$ and we add $i$ to $L_Z(\sigma_Z^{y}(i))$. (Recall that if $\sigma_A^v(i)$ (or $\sigma_Z^v(i)$) is in $B_v^+\setminus B^*_{vw}$, this reachability information was already copied to $\Phi_y$ and that the reachability information for $B_w^-\setminus B_{vw}^*$ was also copied to $\Phi_y$.) Clearly, for each entry $i\in B_y^-$, no entry of $B_y^+\setminus [\sigma_A^v(i),\sigma_Z^v(i)]$ is reachable from $i$. This traversal takes $O(|P_y|+|Q_y|)$ time.

Finally, when we copied information from $B_w^+$ to $B_y^+$, we also copied the lists $L_A$ and $L_Z$ that may include entries of $B^*_{vw}$. Since $B^*_{vw}$ is not a part of the boundary of $B_y$, we need to remove this information from the lists $L_A$ and $L_Z$ of $B_y^+$. We thus go over the entries of $B^*_{vw}$. For each entry $e$ of $B^*_{vw}$, we remove $e$ from $L_A(\sigma^w_A(e))$ and from $L_Z(\sigma^w_Z(e))$. Clearly, this traversal takes $O(|Q_y|)$ time.
\end{proof}

We now show how to use Lemma~\ref{lem:union} to construct $\Gamma$ in $O(m^2)$ time and to update it, when a single entry changes, in $O(m)$ time. We also show how to determine, using $\Gamma$, whether $(m,m)$ is reachable from $(1,1)$ in constant time after the update.

\begin{lemma}
\label{lem:constructGamma}
(a) Given a square matrix $M$, the decomposition tree $\Gamma$ (including the reachability data at its nodes) can be constructed from scratch in $O(m^2)$ time. (b) If a single entry of $M$ is updated, then $\Gamma$ can be updated in $O(m)$ time. (c) Given $\Gamma$, we can determine whether $(m,m)$ is reachable from $(1,1)$ in constant time.
\end{lemma}
\begin{proof}
(a) We construct $\Gamma$ in a bottom-up manner, as prescribed in Lemma~\ref{lem:union}. For the blocks at the leaves, the reachability data is computed in brute force, in $O(1)$ time per block, and at each inner node $y$, the data is computed from the data at its children in time $O(|P_y|+|Q_y|)$, using Lemma~\ref{lem:union}; we refer to $|P_y|+|Q_y|$ as the \emph{size} of the block $B_y$ at $y$. The sizes of the blocks at levels $2j-1$ and $2j$ is $O\left(\dfrac{m}{2^j}\right)$, and the number of these blocks is $O(2^{2j})$. The height of $\Gamma$ is $\lceil \log m \rceil$. The cost of the overall construction of $\Gamma$ is proportional to the sum of the sizes of its blocks (this also holds at the leaf level), which is thus
\begin{equation*}
O\left(\sum_{j=0}^{\lceil \log m\rceil} 2^{2j}\cdot\dfrac{m}{2^j}\right)=O\left(m\sum_{j=0}^{\lceil \log m\rceil} 2^{j}\right)=O\left(m^2\right).
\end{equation*}

\noindent(b) The main observation here is that to update $\Gamma$ when an entry $e$ of $M$ changes, it suffices to update the reachability data along the single path of $\Gamma$ of those nodes $y$ for which $e\in B_y$. (Actually, because of the overlap between block boundaries, there are two such paths that meet at the unique node $y$ for which $e$ belongs to the ``interior'' of the common boundary of the blocks of its children.) The reachability data of the nodes along this path is constructed again from scratch in a bottom-up manner, using Lemma~\ref{lem:union}. The cost of the updates of these blocks is proportional to the sum of their sizes, which is
\begin{equation*}
O\left(\sum_{j=0}^{\lceil \log m\rceil} \dfrac{m}{2^j}\right)=O(m).
\end{equation*}

\noindent(c) To determine whether $(m,m)$ is reachable from $(1,1)$, we simply check in the reachability data structure $\Phi(M)$ of the root of $\Gamma$ whether $(m,m)$ is a $1$-entry that belongs to $[\sigma_A((1,1)), \sigma_Z((1,1))]$ and the flag $f((m,m))$ is true.
\end{proof}

\subsection{A generalized structure for arbitrary matrices}
\label{sec:improved}
We next describe a modified variant of the structure for the case where
$m$ and $n$ are unequal. In what follows we assume, as above and without loss of generality,
that $m\le n$.

We first partition $M$ into $k=O( n/m)$ square blocks $B_1,B_2,\ldots,B_k$,
of size $m\times m$ each  such that consecutive blocks overlap in a single column.
(The last block may be of smaller width, but we
handle it in the same manner as the other blocks; it is easy to show that the bounds of Lemma~\ref{lem:constructGamma} still hold.)
 We build the decomposition tree
and the associated reachability data for each of these blocks, as in Section~\ref{sec:data}; denote the structure for block $B_i$ by $\bar{\Gamma}_i$, for $i=1,\ldots,k$.

We now combine the structures $\bar{\Gamma}_1,\ldots,\bar{\Gamma}_k$ into a single global
structure $\bar{\Gamma}$. For this, we construct a balanced binary tree $T$, with $k$ leaves
 $v_1,\ldots,v_k$, where $v_i$, for $i=1,\ldots,k$, corresponds to $B_i$ and stores $\bar{\Gamma}_i$. Each node $v$ of $T$
represents a block $B_v$ that is the concatenation of the blocks stored at the
leaves of the subtree rooted at $v$. Since each leaf block spans all the rows of $M$,
the common boundary of any pair of consecutive blocks consists only
of a full single column of $M$. The same holds at any node $y$ of $T$, with
left child $v$ and right-child $w$. That is, the common boundary $B_{vw} := B_v^+\cap B_w^-$
between $B_v$ and $B_w$ is vertical, and consists of a full single column of $M$.

We claim that we can merge the reachability structures $\Phi_v$ of $B_v$
and $\Phi_w$ of $B_w$ into the structure $\Phi_y$ of $B_y$
in $O(m)$ time, instead of $O(|B_y|)=O(m+|Q_y|)$ time (as was the cost in the preceding subsection),
which can be much larger. The main observation that facilitates this improvement
is that there is no need to maintain the reachability data $\Phi_v$ at the horizontal
portions of the boundary of any of the blocks $B_v$. This follows from the obvious
property that any path $\pi$ from the initial entry $(1,1)$ to any entry $(i,j)$
in any leaf block reaches $(i,j)$ by crossing all the vertical boundaries
$B_{12},B_{23},\ldots$ that delimit all the preceding leaf blocks, and the
portion $\pi_l$ of $\pi$ within each of the preceding blocks $B_l$ connects an entry
on the left vertical boundary of $B_l$ to an entry on its right vertical boundary.
Note that $\pi_l$ can ``crawl'' along the lower or upper boundary of $B_l$, but to
exit $B_l$ it has to cross the vertical boundary, possibly through its entries in
row $1$ or row $n$.

Figure~\ref{fig:reachable} is an illustration of an inner block $B_y$ of $\bar{\Gamma}$ that is composed of a left block $B_v$ and a right block $B_w$.

\begin{figure}[htb]
\centering\includegraphics[scale=0.6]{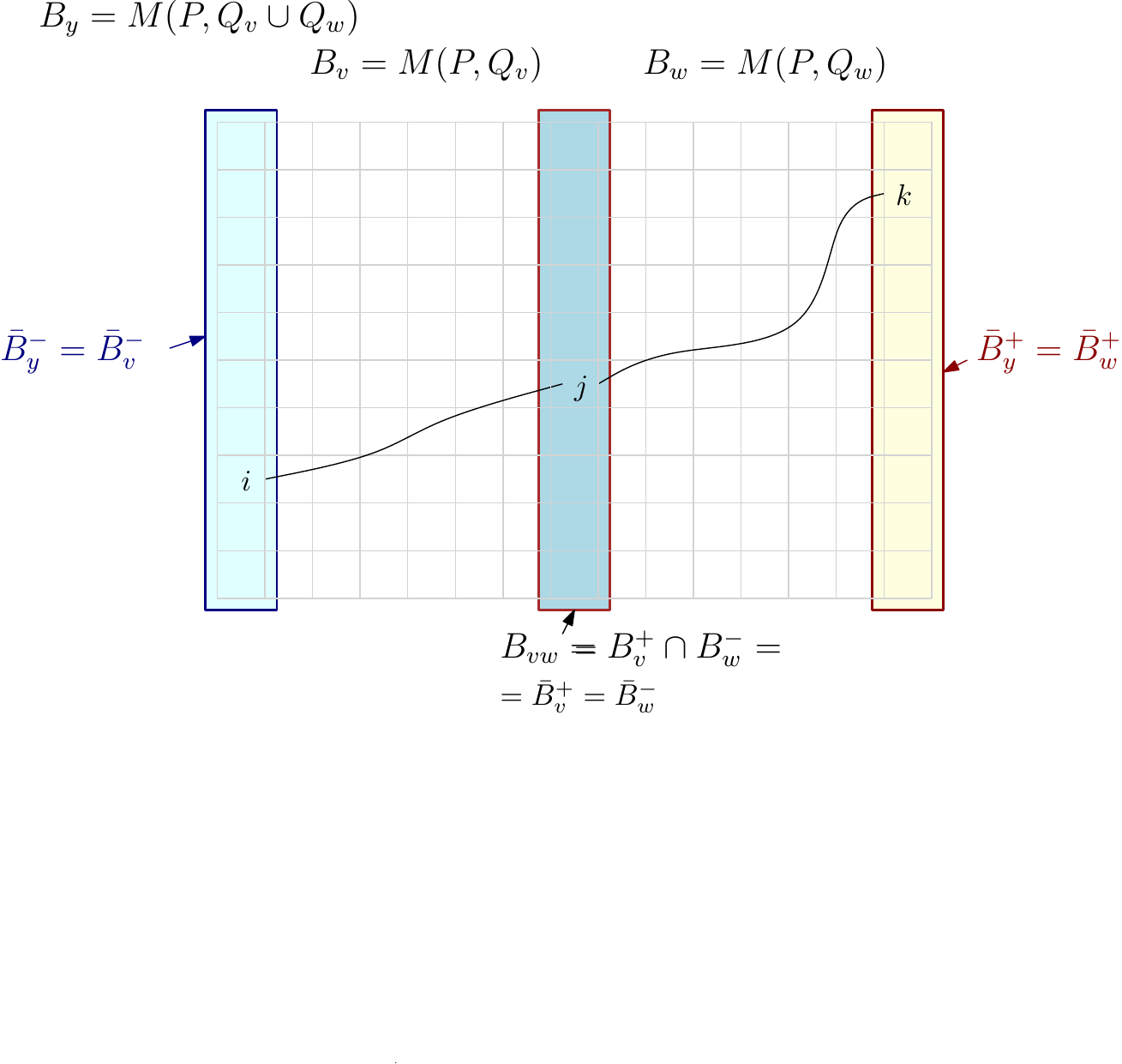} 
\vspace{-2.5cm}
\centering \caption{\small A block $B_y=M(P_y,Q_y)$, corresponding to a node $y$ of $\bar{\Gamma}$, is composed of the blocks $B_v=M(P_v,Q_v)$ and $B_w=M(P_w,Q_w)$ of the children $v,w$ of $y$, with $v$ being the left child and $w$ being the right child. We have $P_y=P_v= P_w=P$ and $Q_y=Q_v \cup Q_w$.
}
\label{fig:reachable}
\end{figure}

We therefore use the same reachability data structure $\Phi_v$ at $v$ as defined
in the previous subsection, except that we limit the input and output domains of
its maps to the vertical boundaries only. Recall our notation from Section~\ref{sec:preliminaries}, where the left (resp., right) vertical boundary of a block $B$ is denoted as $\bar{B}^-$ (resp., $\bar{B}^+$). Specifically, denoting the modified
structure as $\bar{\Phi}_v=\bar{\Phi(B_v)}$, it stores the following items.

\begin{enumerate}
\item
For each $1$-entry $i$ of $\bar{B}_v^-$ we store
\begin{enumerate}
\item
the first entry $\bar{\sigma}_A(i)$ of $\bar{B}_v^+$ that is reachable from $i$, and
\item
the last entry $\bar{\sigma}_Z(i)$ of $\bar{B}_v^+$ that is reachable from $i$.
\end{enumerate}
\item
For each $1$-entry $j$ of $\bar{B}_v^+$ we store
\begin{enumerate}
\item
a flag $\bar{f}(j)$ indicating whether $j$ is reachable from some entry of $\bar{B}_v^-$.
\item
a list $\bar{L}_A(j)$ of the $1$-entries $i\in \bar{B}_v^-$ such that $\bar{\sigma}_A(i)= j$, and
\item
a list $\bar{L}_Z(j)$ of the $1$-entries $i\in \bar{B}_v^-$ such that $\bar{\sigma}_Z(i)= j$.
\end{enumerate}
\end{enumerate}
In other words, $\bar{\Phi}(B_v)$ is a constrained variant of $\Phi(B_v)$, obtained by replacing $B_v^-$ and $B_v^+$ by $\bar{B}_v^-$ and $\bar{B}_v^+$, respectively.
The structure $\bar{\Phi}_v$ of the root $v$ of a child $\bar{\Gamma}_i$ of $\bar{\Gamma}$ is obtained from $\Phi_v$ by first setting, for each entry $i$ of $\bar{B}_v^-$ for which $\sigma_A(i)$ is in $\bar{B}_v^+$ and $\sigma_Z(i)$ is not in $\bar{B}_v^+$,  $\sigma_Z(i)$ to be the last reachable entry $k$ of $\bar{B}_v^+$, then updating $L_Z(k)$ accordingly, and finally ignoring the horizontal parts of the boundaries of $B_v$ and deleting the data regarding them from the lists $L_A$ and $L_Z$ of entries of $\bar{B}_v^+$.

We next claim that the modified structures $\bar{\Phi}_y$ are sufficient for obtaining reachability data for the blocks of $\bar{\Gamma}$, in the precise sense stated below, and that the structure $\bar{\Phi}_y$ at an inner node $y$ of $\bar{\Gamma}$ can be obtained from the structures at the children of $y$ in $O(m)$ time. Concretely, we have the following variants of Lemmas~\ref{lem:linear} and~\ref{lem:union}.

\begin{lemma}\label{lem:linear2}
Given the data structure $\bar{\Phi}(B)$ for a block $B$ of size $r\times c$., and given the entries of $\bar{B}^-$ that are reachable from $(1,1)$, we can determine, in $O(r)$ time, the entries of $\bar{B}^+$ that are reachable from $(1,1)$.
\end{lemma}

\begin{lemma}
\label{cor:Gamma+2}
Let $y$ be an inner node of $\bar{\Gamma}$ with left and right children $v$ and $w$.
Given the reachability data $\bar{\Phi}(M(P_v,Q_v)), \bar{\Phi}(M(P_w,Q_w))$, the data $\bar{\Phi}(M(P_y,Q_y))$ can be computed in $O(m)$ time.
\end{lemma}

The proof is essentially identical to those of Lemma~\ref{lem:linear} and Lemma~\ref{lem:union}, except that we restrict the domains and the images of each of the maps (i.e., $\sigma_A,\sigma_Z,f,L_A$, and $L_Z$) to the vertical portions of the boundaries. This is justified using the observation made earlier that all the reachability paths traverse only vertical boundaries of the relevant blocks --- those that are stored at $\bar{\Gamma}$, from its leaves up, which span the entire range $P$ of rows of $M$. Since we only traverse vertical boundaries, the cost of constructing $\bar{\Phi}_y$ from $\bar{\Phi}_v$ and $\bar{\Phi}_w$ is $O(m)$. $\Box$

The following lemma extends Lemma \ref{lem:constructGamma}.

\begin{lemma}
\label{lem:constructGamma2}
(a) Given the matrix $M$, $\bar{\Gamma}$ can be constructed in $O(mn)$ time. (b) If a single entry of $M$ is updated, then $\bar{\Gamma}$ can be updated in $O(m(1+\log(n/m)))$ time, assuming $m\leq n$. (c) Given $\bar{\Gamma}$, we can determine whether $(m,n)$ is reachable from $(1,1)$ in constant time.
\end{lemma}
\begin{proof}
(a) We construct the structure $\bar{\Gamma}_i$
for each block $B_i$, and extract  $\bar{\Phi}_{v_i}$ from it.
We then construct $\bar{\Phi}_{y}$ for each inner node $y\in T$ by merging the corresponding data structures
of the children of $y$ in $O(m)$ time. We obtain  $\bar{\Gamma}$ at the root of $T$.
Since $T$ is of size $O(n/m)$ and we spend $O(m)$ time at each block,
 it takes $O(n)$ time to construct $\bar{\Gamma}$ from the leaf structures $\bar{\Gamma}_i$ for $i=1,\ldots,k$.
The cost of constructing each $\bar{\Gamma}_i$, $i=1,\ldots,k=\lceil n/m\rceil$, is $O(m^2)$, by Lemma~\ref{lem:constructGamma}, for a total of $O(km^2)=O(mn)$. It follows that
that the overall construction of $\bar{\Gamma}$ takes $O(n+mn)=O(mn)$ time.

\smallskip\noindent(b) To update $\bar{\Gamma}$ when an entry $e$ of $M$ changes
we first need to update the reachability structures along a path in the structure of the block $B_i$ containing $e$ (If $e$ is in the common column of two blocks we update both structures.).
This takes $O(m)$ time by Lemma \ref{lem:constructGamma}.
Once we have the updated $\bar{\Gamma}_i$
we update
 the reachability structures along the  path $\pi$ of $T$ of those nodes $y$ for which $e\in B_y$. (There are two such paths if $e$ is in the common column of two consecutive blocks.)
Since the depth of $T$ is $O(\log(n/m))$ and we spend $O(m)$ time to reconstruct the structure at each node of $T$, we update $T$
in $O(m(1+\log(n/m)))$ time.

\smallskip\noindent(c) As in Lemma~\ref{lem:constructGamma}, to determine whether $(m,n)$ is reachable from $(1,1)$, we simply check in the reachability data structure $\bar{\Phi}(M)$ of the root of $\bar{\Gamma}$ whether $(m,n)\in [\bar{\sigma}_A((1,1)), \bar{\sigma}_Z((1,1))]$ and the flag $\bar{f}((m,n))$ is true.
\end{proof}

\subsection{The overall decision procedure}
\label{sec:overall}
We now put together the pieces of the decision procedure.
We construct the arrangement $\A_\delta$ of the disks $D_\delta(p_i-q_j)$ as in Section~\ref{sec:arrangement} in $O(m^2n^2\log(m+n))$ time. We pick an arbitrary ($0$-, $1$-, or $2$- dimensional) face $f_0$ of $\A_\delta$. $f_0$ corresponds to a unique matrix $M(P,Q+f_0)$ and we construct the data structure $\bar{\Gamma}$ of Section~\ref{sec:improved} based on $M(P,Q+f_0)$. We then perform a traversal of the entire arrangement $\A_\delta$. In each step of the traversal we move from a face $f$ of $\A_\delta$ to a neighbor face $f'$ (both faces are of any dimension $0$, $1$, or $2$). In this step, we either enter a single disk of $\A_\delta$ or exit a single disk of $\A_\delta$. This corresponds to a change in a single entry of $M(P,Q+f)$.
We update $\bar{\Gamma}$ accordingly, in time $O(m(1+\log(n/m)))$, and thereby determine whether $\delta^*(P,Q+f')\leq\delta$. We continue in this manner till we process the entire arrangement. If we encounter a face $f$ along the traversal at which $\delta^*(P, Q+f)\leq \delta$ we report that the minimum distance under translation is $\leq \delta$, and otherwise we report that the minimum distance is $>\delta$.
We thus obtain the following intermediate result.

\begin{theorem}\label{th:decision}
Let $P$, $Q$ be two sequences of points in $\reals^2$ of
sizes $m$ and $n$, respectively and let $\delta >0$ be
a parameter. Then the decision problem, where we want to determine
whether there exists a translation $t\in \mathbb{R}^2$ such that $\dfrechet(P,Q+t) \le \delta$, can be solved in
$O(m^3n^2(1+\log(n/m)))$ time, assuming that $m\le n$.
\end{theorem}

\section{The optimization procedure}
\label{sec:optimization}
We now show how to use the decision procedure of Section~\ref{sec:dynamic} to compute the minimum discrete Fr\'echet distance under translation. 
 Assume without loss of generality that $m\leq n$. As we increase $\delta$, the disks $D_\delta(p_i-q_j)$ expand, and their arrangement $\A_\delta$ varies accordingly. Nevertheless, except for a discrete set of critical values of $\delta$, the combinatorial structure of $\A_\delta$ does not change. That is, the pairs of intersecting disk boundaries remain the same, all their intersection points remain distinct and vary continuously, and no pair of disks are tangent to each other. Consequently, the representation of $\A_\delta$ that we use, namely, a collection of circular sequences of vertices, each containing the vertices of $\A_\delta$ along some circle $C_\delta(p_i-q_j)$, for $p_i\in P$, $q_j\in Q$, sorted along the circle, remain unchanged. The critical values of $\delta$, at which this representation of $\A_\delta$ changes qualitatively, are
\begin{enumerate}
\item\label{enum:critical1}
The radii of the disks that have three points of $P - Q = \{p_i-q_j\mid p_i\in P,\; q_j\in Q\}$ on their boundaries.
\item\label{enum:critical2}
The half-distances between pairs of points of $P-Q$.
\end{enumerate}

There are $O(m^3n^3)$ critical values (most of which are of type \ref{enum:critical1}), so we cannot afford to enumerate them and run an explicit binary search to locate the optimal value of $\delta$ among them.

Instead, we use the parametric searching technique of~\cite{NM83}. In general, using parametric searching can be fairly complicated, since it is based on a simulation of a parallel version of the algorithm. However, we only have to simulate, by a parallel algorithm, the part of the decision procedure that depends on (the unknown value of) the optimum $\delta_T^* = \min_t \dfrechet(P,Q+t)$. In our case, this portion is the construction of $\A_\delta$. 

Instead of actually constructing $\A_\delta$, we first observe that it suffices to restrict our attention to vertices of $\A_\delta$, in the sense that each face $f$ of $\A_\delta$ has a vertex $\xi$, such that all the $1$-entries of $M(P,Q+f)$ are also $1$-entries of $M(P,Q+\xi)$ (the latter matrix can contain additional $1$-entries), so it suffices to test for reachability in the matrices $M(P,Q+\xi)$ associated with vertices $\xi$ of $\A_\delta$. (Technically, we add to the set of vertices one additional point, say the rightmost point, on each disk boundary, to cater to faces that have no real vertices.)

Hence, our parallel implementation of the algorithm will only simulate the construction of the sorted lists of vertices along each of the circles $C_\delta(p_i-q_j)$. Recall that during the parametric searching simulation, we collect comparisons that the decision procedure performs and that depend on $\delta$, and resolve them. This is done by finding the critical values of $\delta$ at which the outcome of some comparison changes, during a single (simulated) parallel step of the algorithm and then by running a binary search through these critical values of $\delta$, guided by the decision procedure of Theorem~\ref{th:decision}. In this manner, we maintain a shrinking half-open interval $I=(\alpha,\beta]$ of values of $\delta$ that contains $\delta_T^*$. Note that we have called the decision procedure at $\alpha$ and it has determined that $\delta_T^*>\alpha$. Then, as is easily seen, $\delta_T^*$ must be at least as large as the first critical value of $\delta$ within $I$ (and it cannot be arbitrarily close to $\alpha$). Assume that we have simulated the construction of $\A_\delta$, and obtained a half-open interval range $I=(\alpha,\beta]$ of $\delta$ that contains $\delta_T^*$. That is, we know that $\alpha<\delta_T^*\leq\beta$, and we know the sorted sequences of vertices of $\A_{\delta_T^*}$ along each circle $C_{\delta_T^*}(p_i-q_j)$. None of the comparisons that the decision procedure has performed has a critical value inside $I$, other than those comparisons that have produced ($\alpha$ and) $\beta$. Hence the output representation of $\A_\delta$ is fixed in the interior of $I$. The rest of the algorithm, which constructs the structure $\bar{\Gamma}$, traverses the vertex sequences along the circles $C_{\delta_T^*}(p_i-q_j)$, and dynamically updates the reachability data, is purely combinatorial, and does not introduce new critical values (i.e., does not involve comparisons that depend on $\delta_T^*$), so there is no need to run it at all. Since the decision procedure fails at $\alpha$ and succeeds at $\beta$, it follows that $\delta_T^*=\beta$.

It is thus sufficient to simulate, at the unknown $\delta_T^*$, an algorithm that
\begin{enumerate}
\item\label{enum:intersection}
finds the intersection points of each circle $C_{\delta_T^*}(p_i - q_j)$ with the circles $\{C_{\delta_T^*}(p_k-q_l)\mid p_k\in P, q_l\in Q\}$, other than itself, and
\item\label{enum:sorting}
sorts, for each circle $C_{\delta_T^*}(p_i-q_j)$, the intersection points that were found on its boundary in step \ref{enum:intersection}, along this boundary.
\end{enumerate}
During the simulation we progressively shrink an interval $I=(\alpha,\beta] \subseteq \mathbb{R}$ that is known to contain $\delta_T^*$. We start with $I = (0,\infty]$.

We first obtain all the $O(m^2n^2)$ critical values of type \ref{enum:critical2}, sort them, and run an explicit  binary search among them guided by the decision procedure. (This part requires no parametric simulation.) As a result $I$ is shrunk to an interval $(\alpha,\beta]$, where $\alpha,\beta$ are two consecutive critical values of type \ref{enum:critical2}. This takes $O(m^2n^2\log(m+n)+m^3n^2(1+\log(n/m))\log(m+n))= O(m^3n^2(1+\log(n/m))\log(m+n))$ time. We can now accomplish step \ref{enum:intersection}, because the property that a pair of circles $C_\delta(p-q), C_\delta(p'-q')$ intersect either holds for all $\delta\in(\alpha,\beta)$ or does not hold for any such $\delta$.

We then execute step \ref{enum:sorting}. The task at hand is to sort, for each circle $C_{\delta_T^*}(p_0-q_0)$, the resulting fixed set of intersection points along $C_{\delta_T^*}(p_0-q_0)$. For each pair $C_{\delta_T^*}(p-q), C_{\delta_T^*}(p'-q')$ of such circles, the order of the intersection points can change only at the radius $\tilde{\delta}$ of the circumcircle of $p_0-q_0,p-q,p'-q'$. We then simulate a parallel sorting procedure, to sort these intersection points along $C_{\delta_T^*}(p_0-q_0)$, and run it in parallel over all these circles. We omit the (by now) routine details of this simulation (see, e.g.,~\cite{AG95} for similar application of parametric searching). They imply that we can simulate this sorting, for each circle $C_{\delta_T^*}(p_i-q_j)$, using $O(mn)$ processors and $O(\log (mn))=O(\log(m+n))$ parallel steps (for a total of $O(m^2n^2)$ processors). Thus, for each parallel step, we need to resolve $O(m^2n^2)$ comparisons, each of which compares $\delta_T^*$ to a critical circumradius of type \ref{enum:critical1}. We run a binary search among these critical values using the decision procedure. This takes $O(m^3n^2(1+\log(n/m))\log(m+n))$ time for each parallel step, for an overall $O(m^3n^2(1+\log(n/m))\log^2(m+n))$ time for $O(\log (m^2n^2))=O(\log(m+n))$ steps.
To (slightly) improve this running time we use the improvement of Cole~\cite{RC87} which finds, for each parallel step, the (weighted) median of the (suitably weighted) unresolved critical values involved in this step, and calls the decision procedure only at this value, instead of using a complete binary search. This allows us to resolve comparisons that contribute at least some fixed fraction of the total weight, while the other unresolved critical values are carried over to the next step with their weights increased. Proceeding in this manner, we make only one call to the decision procedure at each parallel step, and add only $O(\log (m+n))$ parallel steps to the whole procedure. We thus obtain an overall algorithm with $O(m^3n^2(1+\log(n/m))\log(m+n))$ running time.

In conclusion, we get the following main result of the paper.

\begin{theorem} \label{th:optimization}
Let $P$, $Q$ be two sequences of points in $\reals^2$ of respective
sizes $m$ and $n$, where $m\leq n$. Then the minimum discrete Fr\'echet distance under translation between $P$ and
$Q$ can be computed
in $O(m^3n^2(1+\log(n/m))\log(m+n))$ time.
\end{theorem}

\paragraph{Discussion.}
Our algorithm is composed of two main parts. The first part is the construction of the subdivision $\A_\delta$ whose complexity is $O(m^2 n^2)$.
The challenge here is either to argue that, in favorable situations, the actual complexity of $\A_\delta$ is $o(m^2n^2)$, or be able to process only a portion of $\A_\delta$ that has $o(m^2n^2)$ complexity. Here is a simple illustration of such an approach. Consider the case where $P$ and $Q$ are sampled along a pair of \emph{$c$-packed curves}, where a curve $\gamma$ is $c$-packed if, for every disk $D$, the length of $\gamma\cap D$ is at most $c$ times the radius of $D$. Assume also that the sampling is more or less uniform, so that the distance between any pair of consecutive points of $P$ or of $Q$ is roughly some fixed value $\Delta$. We may assume, without loss of generality that $p_1=q_1$. Consider the decision procedure with a given parameter $\delta$, and observe that if $t$ is any translation for which $\delta^*(P,Q+t)\leq \delta$ then $\|t\|\leq \delta$. Therefore, for each $p_i\in P$, the only points $q_j\in Q$ that can align with $p_i$ during a simultaneous traversal of $P$ and $Q+t$, for any such ``good'' translation $t$, are those at distance $\leq 2\delta$ from $p_i$. The assumptions on $P$ and $Q$ imply that the number of such points is at most roughly $2\delta c/\Delta$. That is, instead of constructing the entire arrangement $\A_\delta$, it suffices to construct a coarser arrangement, involving only roughly $2c\delta m/\Delta$ disks. Then, traversing the coarser arrangement is done as before, where each update step (and the following reachability query) cost $O(m\log(1+\log(n/m)))$ time, assuming that $m\leq n$. This improves the running time of the decision procedure to $O\left(\left(\frac{2c\delta}{\Delta}\right)^2 m^3 (1+\log(n/m))\right)$, assuming that $m\leq n$ and $\delta<\frac{\Delta n}{2c}$.

Given this decision procedure, we can solve the optimization problem using parametric searching. However, to ensure that the decision procedure does not become too expensive, we want to run it only with values $\delta=O(\delta^*_T)$. This will become significant only when $\delta^*_T\leq \delta_0:=\frac{n\Delta}{2c}$; otherwise the running time will be close to the running time of the algorithm of Section~\ref{sec:optimization}. 
Therefore, in the following we describe how to solve the optimization problem assuming that $\delta_T^*< \delta_0$ (if the following procedure fails, we run the algorithm of Section~\ref{sec:optimization}).
We also assume, for now, that $\delta_T^*>\Delta$ (we explain below how  the case where $\delta_T^*\leq\Delta$ is dealt with).
We consider the interval $(\Delta, \delta_0)$ that is assumed to contain $\delta_T^*$, and run an ``exponential search'' through it, calling the decision procedure with the values $\delta_i = 2^i \cdot\Delta$, for $i = 0,1,2, \ldots$, in order, until the first time we reach a value $\delta' = \delta_i \geq \delta_T^*$ (and $\delta'<\delta_0$). Note that the cost of running the decision procedure at  $\delta'$ and at $\delta_T^*$ differ by at most a factor of $4$, so the cost of running the decision procedure at $\delta'$ is asymptotically the same as at $\delta_T^*$. Moreover, since the running time bounds on the executions of the decision procedure at $\delta_1,\ldots,\delta_i$ form a geometric sequence, the overall cost of the exponential search is also asymptotically the same as the cost of running the decision procedure at $\delta_T^*$. We then run the parametric searching technique as above, with the constraint that $\delta_T^*$ is at most $\delta'$ (i.e., we set $\delta'$ as the minimal $\beta$ obtained so far). Hence, from now on, each call to the decision procedure made by the parametric searching, will cost no more than the cost of calling the decision procedure with $\delta'$ (which is asymptotically the same as calling the procedure with $\delta_T^*$). We thus obtain an overall algorithm with $O\left(\left(\frac{c\delta_T^*}{\Delta}\right)^2 m^3 (1+\log(n/m))\log(m+n)\right)$ running time.
Note that, in the case where $\delta_T^*\leq\Delta$, after running the decision procedure with $\delta=\Delta$, we realize that $\delta_T^*\leq \Delta$, and run the parametric searching technique with the constraint that $\delta_T^*$ is at most $\Delta$. In this case, the running time of the algorithm is $O\left(c^2 m^3 (1+\log(n/m))\log(m+n)\right)$.



The second part of the algorithm presented in this work is the dynamic data structure for maintaining reachability in $M$. It is an open question of independent interest whether this data structure can be improved. A related problem is whether the techniques used in our structure can be extended to the general case of reachability in planar directed graphs, so as to simplify and improve the efficiency of the earlier competing method of Diks and Sankowski~\cite{DS07}.


\begin{thebibliography}{}

\bibitem{ABKS12} P.~K.~Agarwal, R.~Ben Avraham, H.~Kaplan and
M.~Sharir, Computing the discrete Fr\'echet distance in subquadratic
time, \textit{SIAM J. Comput.} 43(2)
(2014), 429--449, and in arxiv:1204.5333 (2012).


\bibitem{ABB95}
H.~Alt, B.~Behrends and J.~Bl\"{o}mer, Approximate matching of
polygonal shapes, {\it Ann. Math. Artif. Intell.} 13(3-4) (1995),
251--265.

\bibitem{AG95} H.~Alt and M.~Godau, Computing the Fr\'echet distance
between two polygonal curves, \textit{Internat. J. Comput. Geom. Appl.}
5 (1995), 75--91.

\bibitem{AKW01}
H.~Alt, C.~Knauer and C.~Wenk,
Matching polygonal curves with respect to the Fr\'echet distance,
{\it Proc. 18th Annu. Sympos. Theo. Asp. Comput. Sci.}
(2001), 63--74.

\bibitem{Bri14}
K.~Bringmann,
Why walking the dog takes time: Fr\'echet distance has no strongly subquadratic algorithms unless SETH fails,
\textit{Proc. 55th Annu. IEEE Sympos. Found. Comput. Sci.} (2014), 
and in
arxiv:1404.1448 (2014).

\bibitem{BBMM12} K.~Buchin, M.~Buchin, W.~Meulemans and W.~Mulzer,
Four soviets walk the dog --- with an application to Alt's conjecture,
\textit{Proc. 25th Annu. ACM-SIAM Sympos. Discrete Algorithms} (2014), 1399--1413.

\bibitem{RC87}
R.~Cole,
Slowing down sorting networks to obtain faster sorting algorithms,
{\it J. ACM} 34 (1987), 200--208.


\bibitem{DS07}
K.~Diks and P.~Sankowski,
Dynamic plane transitive closure,
{\it Proc. 15th Annu. European Sympos. Algorithms}
(2007), 594--604.

\bibitem{EM94} T.~Eiter and H.~Mannila, Computing discrete Fr\'echet
distance, Technical Report CD-TR 94/64, Christian Doppler Laboratory
for Expert Systems, TU Vienna, Austria, 1994.

\bibitem{JXZ08}
M.~Jiang, Y.~Xu and B.~Zhu,
Protein structure-structure alignment with discrete Fr\'echet distance,
{\it J. Bioinform. Comput. Biol.},
6(1) (2008), 51--64.




\bibitem{NM83}
N.~Megiddo,
Applying parallel computation algorithms in the design of serial
algorithms,
{\it J. ACM} 30 (1983), 852--865.

\bibitem{MC05}
A.~Mosig and M.~Clausen,
Approximately matching polygonal curves with respect to the Fr\'echet distance,
{\it Comput. Geom.}
30(2) (2005), 113--127.

\bibitem{Sub93}
S.~Subramanian,
A fully dynamic data structure for reachability in planar digraphs,
{\it Proc. 1st Annu. European Sympos. Algorithms}
(1993), 372--383.

\bibitem{Wenk02}
C.~Wenk,
{\it Shape matching in higher dimensions},
PhD thesis, Freie Universitaet Berlin
(2002).

\end{thebibliography}
\end{document}